\documentclass[a4paper,twocolumn,11pt,unpublished]{quantumarticle}
\pdfoutput=1
\usepackage[utf8]{inputenc}
\usepackage[T1]{fontenc}
\usepackage{amsfonts,amssymb,amsmath}            
\usepackage[]{graphics,graphicx}            
\usepackage{amsthm}
\usepackage{enumerate,mathtools}
\usepackage{tikz}
\usepackage{adjustbox}
\usetikzlibrary{decorations.pathmorphing,patterns,decorations.markings,matrix}

\usepackage[numbers,sort&compress]{natbib}
    \usepackage{hyperref}
\newtheorem{theorem}{Theorem}

\newtheorem{definition}{Definition}
\newtheorem{lemma}{Lemma}

\newcommand{\ket}[1]{\left | #1 \right\rangle}
\newcommand{\bra}[1]{\left \langle #1 \right |}

\newcommand{\braket}[2]{\left\langle #1|#2\right\rangle}

\renewcommand{\epsilon}{\varepsilon}

\DeclareFontFamily{U}{bbold}{}
\DeclareFontShape{U}{bbold}{m}{n}
 {
  <-5.5> s*[1.069] bbold5
  <5.5-6.5> s*[1.069] bbold6
  <6.5-7.5> s*[1.069] bbold7
  <7.5-8.5> s*[1.069] bbold8
  <8.5-9.5> s*[1.069] bbold9
  <9.5-11> s*[1.069] bbold10
  <11-15> s*[1.069] bbold12
  <15-> s*[1.069] bbold17
 }{}

\DeclareRobustCommand{\identity}{%
  \text{\usefont{U}{bbold}{m}{n}1}%
}

\begin{document}

\title{The Perfect State Transfer Graph Limbo}
\date{\today}

\author{Alastair Kay}
\affiliation{Royal Holloway University of London, Egham, Surrey, TW20 0EX, UK\\\href{mailto:alastair.kay@rhul.ac.uk}{alastair.kay@rhul.ac.uk}}
\orcid{0000-0002-0311-6266}
%
%
\begin{abstract}
Perfect state transfer between qubits on a uniformly coupled network, with interactions specified by a graph, has advantages over an engineered chain, such as much faster transfer times (independent of the distance between the input and output vertices). This is achieved by many couplings working in parallel. The trade-offs seem to be the need for increasing connectivity between qubits, and a large number of vertices in the graph. The size of existing graph constructions scale exponentially in the transfer distance, making these schemes impractical over anything but the shortest distances. This prompts the question of ``How low can you go?'' for the size of the graph achieving a particular transfer distance. In this paper, we present reductions in the vertex count required, although the overall scaling with transfer distance remains exponential. We also tighten existing bounds on the required degree of the vertices of the graph.
\end{abstract}

The task of perfect state transfer \cite{bose2003,christandl2004,kay2010a}, wherein an unknown quantum state $\ket{\psi}$ is transported between two distant locations within a quantum computer by a fixed, time invariant Hamiltonian evolution, was originally intended as a plausible alternative for achieving that task without the vast control overheads of other implementations, such as those derived from the quantum circuit model. A huge range of different assumptions and paradigms have been tried. Two of the major classes of solutions are engineered chains \cite{christandl2004,kay2010a} and uniform networks \cite{christandl2005,kay2011a,godsil2010}. Engineered chains achieve arbitrary transfer distances with a linearly scaling number of qubits, and a linearly scaling transfer time, with the challenge of having to fix the coupling strength between each neighbouring spin to a particular value. There are also some beautiful symmetries (based on the Jordan-Wigner transformation \cite{jordan1928,nielsen2005}) that facilitate a whole host of interesting results, such as \cite{kay2010a,kay2017c,kay2017d}. Uniform networks, meanwhile, still achieve an arbitrary transfer distance, with the advantage of a fixed transfer time but at the cost of requiring more spins. Existing solutions require exponentially many qubits. Can these uniformly coupled networks genuinely be considered as offering a plausible option for state transfer? In this paper, we suggest that the answer is negative for arbitrary transfer distances. By imposing that the coupling network of spins is embeddable in a two- or three-dimensional geometry, the degree of the connectivity is limited. We improve consequential limits on the transfer distance \cite{godsil2010}. We then consider the number of qubits required to achieve transfer over a given distance. While we substantially reduce the exponent compared to existing families of transfer graphs, we are unable to address the question of whether a sub-exponential scaling is possible.

\section{Summary}

Designs for quantum computers are often limited by locality constraints; qubits can only couple to those which are physically close to them. Meanwhile, quantum algorithms assume the ability to couple arbitrary pairs of qubits. To the theorist, there's a simple solution; apply \textsc{swap} gates between neighbouring qubits until the two, initially distantly separated, qubits are moved next to each other so that we can apply a gate between them, and then they can be returned to their initial positions. In practice, however, this is a complicated sequence to coordinate.

The task of perfect state transfer seeks to simplify the task, negating the need to switch on and off many different swap operations in a delicately coordinated fashion, and instead have a single, time invariant operation that will do the job of bringing the two qubits together. It turns out that there are speed advantages in doing this, which yield a corresponding reduction in noise.

In this form, perfect state transfer provides an interesting case study for the broader, and extremely pressing issue of the compilation of quantum circuits -- how should we compile quantum circuits, or at least common, repeated parts of algorithms, so that they run as quickly as possible, with as simple an implementation as possible, on any given quantum hardware? This will maximise what can be achieved on current and near-future quantum computers and quantum technologies. By understanding the whole range of techniques available in one special case, it is hoped that these will translate to a wider range of tasks. There are already some hints in this direction \cite{kay2017c}.

Studies of perfect state transfer have diverged along two main directions. The first is probably more physically relevant, wherein we construct a one-dimensional path between the two qubits, and pass one directly to the other. The time required to achieve the transfer scales with the length of the path, but only requires half the time that the corresponding sequence of \textsc{swap}s would require. Moreover, the other qubits in the chain do not have to be initialised in a fixed initial state. The second, using the mathematical structure of graphs, has many mathematically beautiful features, and a huge speed advantage (requiring a time that is independent of the distance between qubits). However, this speed is generated by using many inter-qubit couplings in parallel. The question is whether these graphs for perfect state transfer ever stand a chance of yielding a physically plausible approach?

In this paper, we use two metrics. The first is the connectivity of each qubit. If a qubit has to couple simultaneously to many other qubits, the scheme is impractical because the whole point is that these schemes are supposed to be working in a system where each qubit has a very limited number of other qubits that it can interact with. The second is the number of qubits required. If we need a large number of qubits, all prepared in the ideal initial state (so we can't even use them in other parts of the computer), this is probably a waste of our resources.

Our results show the impracticality of state transfer graphs on the basis of qubit connectivity. Regarding the size of the graph, we develop new families of graph which significantly reduce the qubit count, but we still suffer from the same exponentially large growth of previous constructions. We have been unable to prove that there are no polynomial sized graphs that achieve perfect transfer. However, we have been able to show that existing lower bounds are inadequate.

\section{Introduction}

Qubits placed at the vertices $V$ of a graph $G$ are assumed to interact via a Hamiltonian that is specified by the edges $E$ of the graph:
$$
H=\frac12\sum_{(n,m)\in E}(X_nX_m+Y_nY_m),
$$
where $X_n$ is the Pauli-$X$ matrix applied to the qubit on vertex $n$. This Hamiltonian obeys the commutation relation
$$
\left[H,\sum_{n\in V}Z_n\right]=0,
$$
demonstrating that it decomposes into subspaces based on the eigenvalues of the total spin operator $\sum_nZ_n$ -- the basis elements $\ket{x}$ for $x\in\{0,1\}^N$ are partitioned by the Hamming weight of $x$. In particular, the one-excitation subspace of $H$ is just represented by the adjacency matrix $A$ of the graph. While other Hamiltonian models have been considered (such as the Heisenberg model), which make connections to other graph descriptions such as the Laplacian or sign-less Laplacian, these are not considered here; many of the techniques cannot be applied.

In the task of perfect state transfer between a pair of qubits $a$ and $b$, an unknown state $\ket{\psi}$ starts on qubit $a$, with all other qubits in the state $\ket{0}$. After evolution by some fixed time $t_0$ under $H$, the state $\ket{\psi}$ is required to arrive on qubit $b$ (up to some known relative phase gate that can be corrected for later). Since $\ket{0}^{\otimes N}$ is an eigenstate of $H$, it suffices to show that
$$
|\bra{a}e^{-iAt_0}\ket{b}|=1
$$
where $\ket{a}$ conveys the presence of a single excitation on qubit $a$, and $0$ everywhere else.
In this context, perfect transfer is achieved not through the engineering of coupling strengths, as is usually the case for one-dimensional chains \cite{bose2003,christandl2004,christandl2005,kay2010a,karbach2005}, but by judicious choice of the graph itself. In many ways, this is far more challenging because there are only discrete choices rather than a continuum of options.

Our aim is to understand when, or even if, perfect transfer on uniform graphs is physically interesting, i.e.\ whether there are realistic experimental prospects. Our main criterion for this is to transfer a state over as large a distance as possible, as simply as possible. For example, given that the path $P_2$ can give a perfect transfer over a distance 1, we see any other graph giving transfer distance 1 (such as integral circulant graphs \cite{angeles-canul2010,basic2013}) as more complicated; they necessarily involve more vertices and a higher degree on at least some vertices, neither of which is desirable. It is not as if such complication can be mitigated by other features such as routing -- after twice the perfect transfer time, there is always a perfect revival on the input site, so the state cannot be passed around multiple different sites, giving different users access to the transferring state.

In this paper, we provide two results. The first is a tightening of a known relationship between the maximum degree of the graph and its transfer distance. This suggests that any realistic graph implementation (determined by local couplings in a two- or three-dimensional structure) will only be able to transfer over very short distances. The second considers the scaling of the number of vertices with the transfer distance. We are able to significantly reduce the best-known exponent on the scaling of the number of vertices, but that relation nevertheless remains an exponential trade-off between transfer distance and vertex number.

\section{The Degree-Distance Bound}

We shall denote the maximum transfer distance of a graph by $D$, its maximum degree by $d$, and its vertex count by $N$. Our aim is, for a particular $D$, to minimise the parameters $d$ and $N$, and see how they might be traded off between each other. The conditions for perfect state transfer are well characterised \cite{godsil2010,kay2011a}, starting from the requirement that perfect state transfer graphs must also exhibit perfect revivals $\ket{a}\xrightarrow{2t_0}\ket{a}$. Indeed many properties have been calculated only based on the existence of perfect revival, and not the further imposition of perfect transfer.

\begin{lemma}
Let $\Phi_a$ be the set of eigenvalues of $A$ for which the eigenvector has support on $a$. Vertex $a$ exhibits perfect revivals iff its eigenvalues $\lambda_n\in\Phi_a$ take the form
\begin{equation}
\lambda_n=\frac12\left(\alpha+\beta_n\sqrt{\Delta}\right), \label{eqn:evalues}
\end{equation}
where $\alpha$, $\{\beta_n\}$ and $\Delta$ are integers. The values $\left\{\frac{\beta_n-\beta_m}{2}\right\}$ are coprime integers $\left(\text{gcd}\left(\left\{ \frac{\beta_n-\beta_m}{2}\right\}\right)=1\right)$. 
\end{lemma}
The most complete proof is in Theorem 6.1 of \cite{godsil2010}, although, there, $\Delta$ is taken to be square-free. We have instead incorporated additional factors inside $\Delta$ in order to fix the GCD condition \cite{coutinho2018}.

\begin{lemma}
In order to achieve perfect state transfer between qubits $a$ and $b$, a graph exhibiting perfect revivals must satisfy
$$
\braket{\lambda_n}{a}=\pm\braket{\lambda_n}{b}
$$
for all $\lambda_n\in\Phi_a$, and the parity of the integer $\frac{\beta_n-\beta_m}{2}$ must match the sign $\frac{\braket{\lambda_m}{a}\braket{\lambda_n}{b}}{\braket{\lambda_n}{a}\braket{\lambda_m}{b}}$. Perfect transfer occurs at all times that are odd multiples of $\pi/\sqrt{\Delta}$, with perfect revivals on the input at all even multiples. 
\end{lemma}

The maximum degree of a graph is a particularly limiting parameter for physical realisation. If we imagine a physical realisation in which only local spins can interact directly, then embedding in a two or three-dimensional geometry places severe constraints on the number of vertices that can interact; on the order of 4 in two-dimensions or 6 in three-dimensions. This in turn constrains the distance of transfer:
\begin{lemma}\label{lem:degdist}
For a given $\Delta$, $D\leq 2d/\sqrt{\Delta}.$
\end{lemma}
\begin{proof}
This originates from \cite{godsil2010}. To summarise in our slightly different notation, apply Gershgorin's circle theorem to the adjacency matrix $A$: each diagonal element is 0, and the off-diagonals total to no more than $d$, so all eigenvalues are bounded between $\pm d$. Since every distinct eigenvalue must be separated by at least $\sqrt{\Delta}$, there are no more than $2d/\sqrt{\Delta}+1$ distinct eigenvalues, and $|\Phi_a|\leq 2d/\sqrt{\Delta}+1$. \cite{coutinho2018} makes the argument that if the eccentricity of $a$, the maximum distance of any vertex from $a$, is $\epsilon_a$, then $\{A^k\ket{a}\}_{k=0}^{\epsilon_a}$ must form a basis, which can be no larger than $|\Phi_a|$. Thus, $D\leq \epsilon_a\leq|\Phi_a|-1\leq 2d/\sqrt{\Delta}$.
\end{proof}

\begin{figure}
\begin{center}
\includegraphics[width=0.45\textwidth]{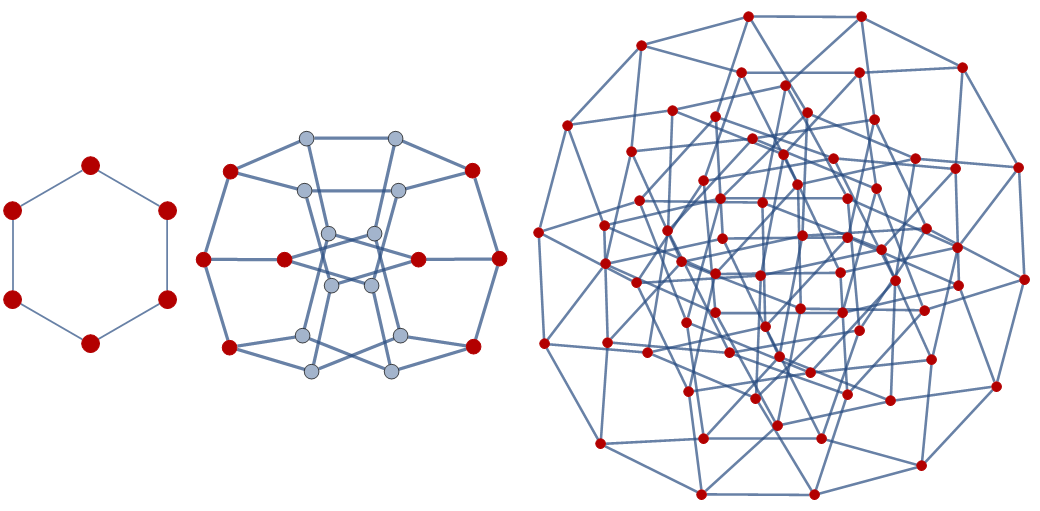}
\end{center}
\caption{Regular graphs of degree $d=2,3,4$ where vertices in red have eccentricities $2d-1$ and are periodic.}\label{fig:examples}
\end{figure}

The above proof applies to all graphs with perfect revivals, and there are no known perfect transfer graphs for which $D>d$, while there are examples of graphs with perfect revivals that exceed that bound. Some examples are depicted in Fig.\ \ref{fig:examples}. Our first task is to tighten this bound by specifically imposing the properties of perfect transfer graphs instead of just perfect revival graphs. However, we shall do this in a slightly restricted setting, wherein we assume that the vertex $a$ that we transfer from is {\em spectrally extremal}, i.e.\ the number of eigenvectors with non-trivial support on $a$ is $|\Phi_a|=\epsilon_a+1=D+1$. This would seem to be a reasonable assumption. To make progress in excluding certain values of $\Delta$, we need an extra result first:
\begin{lemma} \label{lem:helper}
For an ordered set of integers, $\Lambda$, of size  $|\Lambda|=N>3$ such that the $n^{th}$ value has parity $(-1)^{n+1}$, $R$ is rational, and the denominator of $R$ (in its simplest possible form) contains a factor of 2, where
$$
R=\frac{1}{\displaystyle\sum_{\lambda\in\Lambda}\frac{(-1)^{\lambda}}{\prod_{\mu\in\Lambda\setminus\lambda}(\lambda-\mu)}}.
$$
\end{lemma}
\begin{proof}
As it is only the eigenvalue differences that are relevant, we take the smallest value in $\Lambda$ to be 0, and the largest to be $\lambda_{\max}$. Let $\bar\Lambda$ contain all the integers between 0 and $\lambda_{\max}$ that are not contained within $\Lambda$. It is vital to note that these occur in consecutive pairs to maintain the parity condition. We have that
$$
R=\frac{1}{\sum_{\lambda\in\Lambda}\frac{\prod_{\mu\in\bar\Lambda}(\mu-\lambda)}{\lambda!(\lambda_{\max}-\lambda)!}}.
$$
If it were the case that the members of $\bar\Lambda$ did not occur in consecutive pairs, there would be an additional sign on each term in the sum given by $(-1)^{|\{\mu\in\bar\Lambda:\mu<\lambda\}|}$, but we avoid this complication which would invalidate the rest of the proof. The sum can now be extended from $\Lambda$ to the entire set of integers 0 to $\lambda_{\max}$ since the additional terms from $\bar\Lambda$ are all multiplied by zero. Thus,
$$
R=\displaystyle\frac{1}{\sum_{n=0}^{\lambda_{\max}}\frac{\prod_{\mu\in\bar\Lambda}(\mu-n)}{n!(\lambda_{\max}-n)!}}.
$$
We rewrite this as
$$
R=\displaystyle\frac{\lambda_{\max}!}{\sum_{n=0}^{\lambda_{\max}}\binom{\lambda_{\max}}{n}\prod_{\mu\in\bar\Lambda}(\mu-n)}.
$$
Next, we express $\prod_{\mu\in\bar\Lambda}(\mu-n)$ in the form
$$
\sum_{q=0}^{|\bar\Lambda|}f^{(q)}(\bar\Lambda)\frac{n!}{(n-q)!}
$$
The precise form of $f^{(q)}(\bar\Lambda)$ is irrelevant, but it follows by induction (adding pairs of integers to $\bar\Lambda$ of $\mu$ and $\mu+1$) that these are integers. This lets us write
$$
R=\displaystyle\frac{\lambda_{\max}!}{\sum_{q=0}^{|\bar\Lambda|}f^{(q)}(\bar\Lambda)\frac{\lambda_{\max}!}{(\lambda_{\max}-q)!}\sum_{n=0}^{\lambda_{\max}-q}\binom{\lambda_{\max}-q}{n}}
$$
which reduces to
\begin{equation}\label{eqn:R}
R=\displaystyle\frac{1}{\sum_{q=0}^{|\bar\Lambda|}f^{(q)}(\bar\Lambda)\frac{2^{\lambda_{\max}-q}}{(\lambda_{\max}-q)!}}.
\end{equation}

Consider each term $2^x/x!$. How many factors of 2 does $x!$ contain? There are $\lfloor x/2\rfloor$ even numbers, $\lfloor x/4\rfloor$ that are divisible by 4 (each contributing an extra factor of 2), $\lfloor x/8\rfloor$ that are divisible by 8, and so on. Hence, the number of factors of 2 is no more than
$$
\frac{x}{2}+\frac{x}{4}+\ldots +1= \frac{x}{2}\sum_{n=0}^\infty 2^{-n}-\sum_{n=1}^\infty 2^{-n}=x-1,
$$
so there are no more than $x-1$ factors of 2, and not all of the terms in $2^{x}$ cancel. Hence, every term in the sum over $q$ in Equation (\ref{eqn:R}) has an even numerator. Thus, $R$ has a denominator containing a factor of 2.
\end{proof}

\begin{theorem}
For spectrally extremal perfect state transfer, $\Delta$ must be even, and if $\Delta\text{ mod }4=2$, the transfer distance $D$ must be even.
\end{theorem}
\begin{proof}
Due to Theorem 11 in \cite{coutinho2016a}, the eigenvector elements satisfy
$$
|\braket{a}{\lambda}|^2=\frac{k}{\left|\prod_{\mu\in\Phi_a\setminus\lambda}(\lambda-\mu)\right|}
$$
for all $\lambda\in\Phi_a$, where $k$ is the number of paths of length $\epsilon_a$ from $a$ to $b$ (deriving from the fact that $\bra{a}H^n\ket{b}=0$ if $n<\epsilon_a$ and $\bra{a}H^{\epsilon_a}\ket{b}=k$, and imposing strong cospectrality: $\braket{\lambda}{a}=\pm\braket{\lambda}{b}$). Moreover, the ordered values $\lambda-\mu$ must be alternately odd and even multiples of $\sqrt{\Delta}$. If we impose that $\sum_{\lambda\in\Phi_a}|\braket{a}{\lambda}|^2=1$, and take the standard form of the eigenvalues, Eq.\ (\ref{eqn:evalues}), then
$$
k=\frac{\Delta^{\frac{|\Phi_a|-1}{2}}}{\left|\sum_{n}\frac{(-1)^{\lambda}}{\prod_{m\neq n}\left(\frac{\beta_n-\beta_m}{2}\right)}\right|},
$$
which must be an integer. We know by Lemma \ref{lem:helper} that the right-hand side is of the form $\Delta^{\frac{|\Phi_a|-1}{2}}R$ where $R$ has a factor of 2 in the denominator (recall, it is the values $\frac{\beta_n-\beta_m}{2}$ that must be alternately odd and even integers). In order to cancel that factor of 2, it must be created by the $\Delta$ term. Either $\Delta$ is divisible by four, or it's divisible by two and the transfer distance is even.
\end{proof}

By eliminating the $\Delta=1$ case, we have tightened the degree-distance relation provided in Lemma \ref{lem:degdist}, putting a physically realistic option further from reach. The family of perfect state transfer graphs created from the hypercube of $P_2$ (of dimension $D$) saturate this bound: $\Delta=4$ and $D=d$. This family requires a number of vertices $N=2^D$. Better performance can be found for the family of graphs created from the hypercube of $P_3$, but only for even distances since $\Delta=2$: $D=d$ and $N=3^{D/2}$. We do not know if it is possible to saturate the $D\leq \sqrt{2}d$ bound when $\Delta=2$, or whether the true bound is indeed $D\leq d$ (we suspect the latter).

Note that in the following section, particular emphasis will be placed on the `standard' solutions for engineered perfect state transfer chains \cite{christandl2004}: for extremal transfer where $|\Phi_a|=D+1$ and $\Delta=4$,  all systems have a graph quotient corresponding to some engineered perfect state transfer chain, and this is the only chain that can saturate the $D=d$ limit.

\section{Vertex Number}\label{sec:vertex}

We now focus on minimising the total number of vertices in a graph of fixed transfer distance $D$. An important result here, originally due to \cite{christandl2005}, is that if two graphs $G_1$ and $G_2$ achieve perfect transfer in the same time (i.e.\ have the same value of $\Delta$) over distances $D_1$ and $D_2$, then their Cartesian product $G_1\square G_2$ (adjacency matrix $A_1\otimes\identity+\identity\otimes A_2$) achieves perfect transfer in that time, over a distance $D_1+D_2$. Hence, any graph $G$ that achieves perfect transfer over distance $D$ determines a family $G^{\square k}$ with transfer distance $kD$ (and a number of vertices $N^k$). Our measure of success must, therefore, be the exponent. We refer to this as the efficiency, $\eta$:
$$
\eta=\frac{1}{D}\log_2(N).
$$
As a point of reference, the $P_2$ hypercubes all have $\eta=1$, and the $P_3$ hypercubes have $\eta=\frac{1}{2}\log_2(3)\approx 0.792$. Smaller is better. Prior to this work, this was the best known efficiency. Our main tool in improving this efficiency is the concept of the partitioned graph and how, for a fixed partitioning, there is a manipulation rule that can trade number of vertices for number of edges without changing the ability for state transfer.

\subsection{The Partitioned Graph}

\begin{definition}
An equitable distance partition of a graph $G$ comprises distinct sets of vertices $\{V_i\}$ (called nodes) such that
\begin{itemize}
	\item Node $V_0$ is a single vertex; the input vertex.
	\item All vertices in $V_i$ are equidistant from $V_0$.
	\item Every vertex in node $V_i$ connects to the same number of vertices in $V_j$.
	\item No edges join vertices in the same node.  
\end{itemize}
\end{definition}

The `Partitioned Graph' is an equitable distance partition of the graph. In this, we partition vertices into equivalent sets, which we call a {\em node}. The following diagrammatic representation specifies a graph in which one set of $N_1$ vertices, each with degree $d_1$, connects to vertices in the other set, of which there are $N_2$, each with degree $d_2$:
\begin{center}
\begin{tikzpicture}
\draw [thick] (0,0) node[circle,style={fill=black,minimum width=1,text=white}]{$N_1$} -- (2,0) node[circle,style={fill=black,minimum width=1,text=white}]{$N_2$};
\draw (0.6,0.2) node{$d_1$};
\draw (1.4,0.2) node{$d_2$};
\end{tikzpicture}
\end{center}
A required consistency condition is that $N_1d_1=N_2d_2$. There is also a restriction that $d_1\leq N_2$ and $d_2\leq N_1$.

Partitioned graphs have a natural graph quotient, produced by working in the subspace comprising the uniform superposition across all vertices within an individual node. Each node represents a vertex in the quotient graph, and each edge has an effective coupling strength $\sqrt{d_1d_2}$. For perfect state transfer, we identify the input and output vertices as separate nodes (with occupancy 1). Perfect transfer between these two vertices in the graph quotient implies perfect transfer in the original graph.

\subsection{A Node-Preserving Manipulation Rule for Partitioned Graphs} \label{sec:manipulate}

There is a particularly simple manipulation rule for partitioned graphs, which preserves the partitioning. For any node with an occupancy $N_1$ that is divisible by 4, such that every neighbouring node has a corresponding degree that is even, we can perform the replacement depicted in Fig.\ \ref{fig:manipulate}. The consistency condition is still fulfilled, and the effective coupling in the quotient graph is unchanged. As such, the ability of the graph to perform state transfer between the extremal vertices is unchanged.
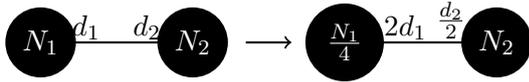
\begin{figure}[t]
\begin{center}
\begin{tikzpicture}
\draw [thick] (0,0) node[circle,style={fill=black,minimum width=1,text=white}]{$N_1$} -- (2,0) node[circle,style={fill=black,minimum width=1,text=white}]{$N_2$};
\draw (0.6,0.2) node{$d_1$};
\draw (1.4,0.2) node{$d_2$};

\draw [thick,->] (2.7,0) -- (3.3,0);

\draw [thick] (4,0) node[circle,style={fill=black,minimum width=1,text=white}]{$\frac{N_1}{4}$} -- (6,0) node[circle,style={fill=black,minimum width=1,text=white}]{$N_2$};
\draw (4.8,0.2) node{$2d_1$};
\draw (5.4,0.3) node{$\frac{d_2}{2}$};

\end{tikzpicture}
\end{center}
\caption{Basic manipulation rule for the partitioned graph, preserving the property of state transfer. Requires $N_1$ divisible by 4, $d_2$ even and $2d_1\leq N_2$.}\label{fig:manipulate}
\end{figure}

There are several variants of this manipulation rule. For example, one could reduce $N_1$ by any square factor $n^2$ so long as the degrees on the adjoining nodes are divisible by $n$. However, in practice we have found the $n=2$ is by far the most prevalent. Equally, while we usually think about the manipulation rule for the purposes of decreasing the vertex count, there are instances where increasing it is beneficial. For example,
\begin{align*}
\begin{tikzpicture}[baseline={([yshift=-.7ex]current bounding box.center)},scale=0.7, every node/.style={scale=0.7}]
\draw [thick] (0,0) node[circle,style={fill=black,minimum width=0.8cm,text=white}]{$1$} -- (2,0) node[circle,style={fill=black,minimum width=0.8cm,text=white}]{$16$};
\draw (0.6,0.2) node{$16$};
\draw (1.4,0.2) node{$1$};
\end{tikzpicture} & \rightarrow \begin{tikzpicture}[baseline={([yshift=-.7ex]current bounding box.center)},scale=0.7, every node/.style={scale=0.7}]
\draw [thick] (0,0) node[circle,style={fill=black,minimum width=0.8cm,text=white}]{$4$} -- (2,0) node[circle,style={fill=black,minimum width=0.8cm,text=white}]{$16$};
\draw (0.6,0.2) node{$8$};
\draw (1.4,0.2) node{$2$};
\end{tikzpicture}	\\
&\rightarrow\begin{tikzpicture}[baseline={([yshift=-.7ex]current bounding box.center)},scale=0.7, every node/.style={scale=0.7}]
\draw [thick] (0,0) node[circle,style={fill=black,minimum width=0.8cm,text=white}]{$4$} -- (2,0) node[circle,style={fill=black,minimum width=0.8cm,text=white}]{$4$};
\draw (0.6,0.2) node{$4$};
\draw (1.4,0.2) node{$4$};
\end{tikzpicture}
\end{align*}
We could not, initially, apply the rule on the second node because the degree on the first node was too high, but by first applying the manipulation rule in reverse on the first node, this problem was circumvented, and there was a net reduction in the vertex count from 17 to 8.

Lastly, this manipulation rule can be applied to subgraphs of nodes -- the occupancy of every node in the subgraph is reduced by the factor of 4, while all internal edges retain the same degrees, and edges joining the subgraph to the rest of the graph have degrees that change in the same way as for the basic rule. This allows us to significantly reduce the conditions under which the rule can be applied -- if the occupation of a node is divisible by 4, then  since $N_id_i$ is divisible by 4, each of the neighbouring nodes has an occupation that is divisible by 4 (in which case, we can expand the subgraph to include that node) or the degree of the connected node is divisible by 2, and hence the basic version of the rule is applicable.

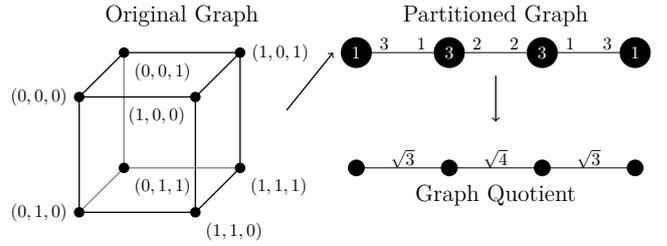
\begin{figure}
\begin{center}
\begin{adjustbox}{width=0.5\textwidth}
\begin{tikzpicture}
\draw (1.25,3.3) node{\Large{Original Graph}};
\draw (8,3.3) node{\Large{Partitioned Graph}};
\draw (8,-0.6) node{\Large{Graph Quotient}};
  \draw[thick](2.5,2.5,2.5)--(2.5,0,2.5);
  \draw[gray,thick](2.5,0,0)--(0,0,0)--(0,2.5,0);
  \draw[gray,thick](0,0,0)--(0,0,2.5);
  \draw[thick] (2.5,2.5,0) node[inner sep=2pt,circle,draw,fill,label={right:$(1,0,1)$}]{}  -- (0,2.5,0) node[inner sep=2pt,circle,draw,fill,label={south east:$(0,0,1)$}]{} -- (0,2.5,2.5) node[inner sep=2pt,circle,draw,fill,label={left:$(0,0,0)$}]{} -- (2.5,2.5,2.5) node[inner sep=2pt,circle,draw,fill,label={south west:$(1,0,0)$}]{}  -- (2.5,2.5,0) -- (2.5,0,0) node[inner sep=2pt,circle,draw,fill,label={south east:$(1,1,1)$}]{}  -- (2.5,0,2.5) node[inner sep=2pt,circle,draw,fill,label={south east:$(1,1,0)$}]{}  -- (0,0,2.5) node[inner sep=2pt,circle,draw,fill,label={left:$(0,1,0)$}]{} -- (0,2.5,2.5);
  \draw[thick] (0,0,0) node[inner sep=2pt,circle,draw,fill,label={south east:$(0,1,1)$}]{};

  \draw[thick,->] (3.5,1.25) -- (4.5,2.5);

\draw (5,2.5) node[circle,style={fill=black,minimum width=1,text=white}]{$1$} -- (7,2.5) node[circle,style={fill=black,minimum width=1,text=white}]{$3$} -- (9,2.5) node[circle,style={fill=black,minimum width=1,text=white}]{$3$} -- (11,2.5) node[circle,style={fill=black,minimum width=1,text=white}]{$1$};
\draw (5.6,2.7) node{$3$};
\draw (6.4,2.7) node{$1$};
\draw (8.4,2.7) node{$2$};
\draw (7.6,2.7) node{$2$};
\draw (10.4,2.7) node{$3$};
\draw (9.6,2.7) node{$1$};

\draw [thick,->] (8,2) -- (8,1);
\draw (5,0) node[circle,style={fill=black,minimum width=0.5,text=white}]{} -- (7,0) node[circle,style={fill=black,minimum width=0.5,text=white}]{} -- (9,0) node[circle,style={fill=black,minimum width=0.5,text=white}]{} -- (11,0) node[circle,style={fill=black,minimum width=0.5,text=white}]{};
\draw (6,0.2) node{$\sqrt{3}$};
\draw (8,0.2) node{$\sqrt{4}$};
\draw (10,0.2) node{$\sqrt{3}$};
\end{tikzpicture}
\end{adjustbox}
\end{center}
\caption{The three-fold hypercube of the two-vertex path $P_2$. The graph partition is simply based on the distance of each vertex from the input vertex. This yields a graph quotient which is the same as the standard perfect state transfer solution.}\label{fig:P2D3}
\end{figure}

For example, consider the $D$-dimensional hypercube of $P_2$. Every vertex at a given distance from the input vertex is equivalent, and can thus be partitioned into a single node (the distance partition). One thus has a chain-like structure of $D+1$ nodes, where the node occupancies are $\binom{D}{n}$ for $n=0,\ldots,D$. The quotient graph is exactly the standard perfect state transfer chain with engineered couplings $\sqrt{(n+1)(D-n)}$. This is exactly the construction originally performed in \cite{christandl2005}. The special case of $D=3$ is depicted in Fig.\ \ref{fig:P2D3}. While this example does not demonstrate any gain under the manipulation rules, $D=6$ does:
\begin{center}
\begin{adjustbox}{width=0.5\textwidth}
\begin{tikzpicture}
\draw (0,0) node[circle,style={fill=black,minimum width=0.8cm,text=white}]{$1$} -- (2,0) node[circle,style={fill=black,minimum width=0.8cm,text=white}]{$6$} -- (4,0) node[circle,style={fill=black,minimum width=0.8cm,text=white}]{$15$} -- (6,0) node[circle,style={fill=black,minimum width=0.8cm,text=white}]{$20$} -- (8,0) node[circle,style={fill=black,minimum width=0.8cm,text=white}]{$15$} -- (10,0) node[circle,style={fill=black,minimum width=0.8cm,text=white}]{$6$} -- (12,0) node[circle,style={fill=black,minimum width=0.8cm,text=white}]{$1$};
\draw (0.6,0.2) node{$6$};
\draw (1.4,0.2) node{$1$};
\draw (2.6,0.2) node{$5$};
\draw (3.4,0.2) node{$2$};
\draw (4.6,0.2) node{$4$};
\draw (5.4,0.2) node{$3$};
\draw (11.4,0.2) node{$6$};
\draw (10.6,0.2) node{$1$};
\draw (9.4,0.2) node{$5$};
\draw (8.6,0.2) node{$2$};
\draw (7.4,0.2) node{$4$};
\draw (6.6,0.2) node{$3$};

\draw [thick,->] (6,-0.7) -- (6,-1.3);

\draw (0,-2) node[circle,style={fill=black,minimum width=0.8cm,text=white}]{$1$} -- (2,-2) node[circle,style={fill=black,minimum width=0.8cm,text=white}]{$6$} -- (4,-2) node[circle,style={fill=black,minimum width=0.8cm,text=white}]{$15$} -- (6,-2) node[circle,style={fill=black,minimum width=0.8cm,text=white}]{$5$} -- (8,-2) node[circle,style={fill=black,minimum width=0.8cm,text=white}]{$15$} -- (10,-2) node[circle,style={fill=black,minimum width=0.8cm,text=white}]{$6$} -- (12,-2) node[circle,style={fill=black,minimum width=0.8cm,text=white}]{$1$};
\draw (0.6,-1.8) node{$6$};
\draw (1.4,-1.8) node{$1$};
\draw (2.6,-1.8) node{$5$};
\draw (3.4,-1.8) node{$2$};
\draw (4.6,-1.8) node{$2$};
\draw (5.4,-1.8) node{$6$};
\draw (11.4,-1.8) node{$6$};
\draw (10.6,-1.8) node{$1$};
\draw (9.4,-1.8) node{$5$};
\draw (8.6,-1.8) node{$2$};
\draw (7.4,-1.8) node{$2$};
\draw (6.6,-1.8) node{$6$};
\end{tikzpicture}
\end{adjustbox}
\end{center}
The most impressive gains that we found for this construction are for $D=16$, where we reduce from $2^{16}$ vertices to just 8874 ($\eta=0.820$). Sec.\ \ref{sec:P3} will demonstrate further gains by using different partitions.

\subsection{The \texorpdfstring{$\Delta$}{Delta}-doubling Lift}

In order to reduce the number of vertices as much as possible, it seems sensible to start from the families with smallest total vertex number: the $P_3$ hypercubes. However, since these have $\Delta=2$, they can only transfer over even distances. If we want odd distances, then we need to move to $\Delta=4$. Thankfully, there is a connection. We can move from any graph with a bipartite node structure and a given value of $\Delta$ to a graph with a value of $2\Delta$, using the same node structure.

If the partitioned structure is bipartite, then to every node we can assign a label 0 or 1 based on the parity of distance from the input vertex. The input and output vertices can be labelled 0 without loss of generality. Our $\Delta$-doubling lift simply changes the node occupancy of a node with label 1 by doubling it, while degrees on the nodes with label 0 are doubled. This maintains the consistency condition while the effect in the graph quotient is to change every edge by a factor of $\sqrt{2}$. Hence all the eigenvalues are multiplied by $\sqrt{2}$, and $\Delta$ is thus doubled. For an example, see Fig.\ \ref{fig:doubledown}. This is very similar to a trick used in \cite{kay2005}.

\begin{figure}
\begin{center}
\begin{adjustbox}{width=0.48\textwidth}
\begin{tikzpicture}
\draw (0,0) node[circle,style={fill=black,minimum width=0.8cm,text=white}]{$1$} -- (2,0) node[circle,style={fill=black,minimum width=0.8cm,text=white}]{$1$} -- (4,0) node[circle,style={fill=black,minimum width=0.8cm,text=white}]{$1$};
\draw (0.6,0.2) node{$1$};
\draw (1.4,0.2) node{$1$};
\draw (2.6,0.2) node{$1$};
\draw (3.4,0.2) node{$1$};

\draw [thick,->] (5,0) -- (6,0);

\draw (7,0) node[circle,style={fill=black,minimum width=0.8cm,text=white}]{$1$} -- (9,0) node[circle,style={fill=black,minimum width=0.8cm,text=white}]{$2$} -- (11,0) node[circle,style={fill=black,minimum width=0.8cm,text=white}]{$1$};
\draw (7.6,0.2) node{$2$};
\draw (8.4,0.2) node{$1$};
\draw (9.6,0.2) node{$1$};
\draw (10.4,0.2) node{$2$};
\end{tikzpicture}
\end{adjustbox}
\end{center}
\caption{The path $P_3$ has $\Delta=2$. Applying the $\Delta$-doubling lift produces $P_2^{\otimes 2}$ ($\Delta=4$).}
\label{fig:doubledown}
\end{figure}
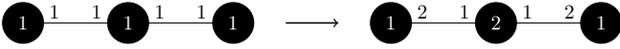

It must also be emphasised that the $\Delta$-doubling lift is highly compatible with the manipulation rule of Sec.\ \ref{sec:manipulate} for further reductions. For nodes with the 1 label, the neighbouring degrees are even due to the construction, and the node occupancy is already even, so it only needs an extra factor of 2 to be applicable (modulo some additional care with the conditions on the degree).

\subsection{Node Splitting}

While we have specified one manipulation rule that preserves the node structure, other manipulation rules can be formulated which, for example, show how a single node can be broken into several smaller ones, under limited conditions. (Technically, the method we are about to present contains the previous manipulation rule.)

Consider a node with occupation $N$, and $p$ neighbouring nodes with occupations $M_j$ for $j=1,\ldots,p$, and coupled (in the graph quotient) by coupling strength $J_j$. We split the single node into a set of $k$ nodes (the value of $k$ to be determined), each with occupancy $N_i$. The degrees of each node when connecting to the $p$ neighbours are $\{d^{i}_j\}$. Our consistency condition is as before: $N_id^{i}_j$ should be divisible by $M_j$, and the effective coupling strength of the graph quotient is correctly maintained if
$$
M_jJ_j^2=\sum_iN_i{d_j^{i}}^2
$$
We also need to ensure that a uniform superposition over all of the vertices in any of the neighbouring nodes hops onto the same state across the new set of nodes. This is achieved by ensuring that $$d^{i}_j=\alpha_jd^{i}_1$$ for all $i,j$. In turn, this means
$$
\alpha_j=\sqrt{\frac{M_kJ_j^2}{M_1J_1^2}}.
$$
This number is rational, and we refer to it in its lowest possible form as $\alpha_j=\frac{n_j}{d_j}$. For this to work, $d^{i}_1$ must be divisible by all the $d_j$. Also, to achieve the consistency condition that $N_id^{i}_j$ should be divisible by $M_j$, we need $N_id^{i}_1$ to be divisible by $M_jd_j/\text{GCD}(M_j,n_j)$ for all $j=2,3,\ldots,p$ (as well as $M_1$). Subject to these constraints, our aim is to minimise $\sum_iN_i$.

The simplest example of this in practice is the $D=4$ transfer case, starting from the $P_2$ hypercube construction. We have the partitioned graph:
\begin{center}
\begin{adjustbox}{width=0.3\textwidth}
\begin{tikzpicture}
\draw (0,0) node[circle,style={fill=black,minimum width=1,text=white}]{$1$} -- (2,0) node[circle,style={fill=black,minimum width=1,text=white}]{$4$} -- (4,0) node[circle,style={fill=black,minimum width=1,text=white}]{$6$} -- (6,0) node[circle,style={fill=black,minimum width=1,text=white}]{$4$} -- (8,0) node[circle,style={fill=black,minimum width=1,text=white}]{$1$};
\draw (1,0.2) node{$\sqrt{4}$};
\draw (3,0.2) node{$\sqrt{6}$};
\draw (5,0.2) node{$\sqrt{6}$};
\draw (7,0.2) node{$\sqrt{4}$};
\end{tikzpicture}
\end{adjustbox}
\end{center}
We aim to replace the central node with a set of $k$ nodes. By symmetry, we can take $p=1$. Thus, we only have to find $\min\sum_iN_i$ subject to $\sum_iN_i{d^{i}}^2=24$ and $N_id^{i}$ divisible by 4. Of course, the original values still work ($k=1, N_1=6, d^1=2$), but there is another solution: $k=2, N_1=1, d^1=4, N_2=2, d^2=2$. This reduces the vertex count ($N_1+N_2=3<6$), and returns the $D=4$ graph of 13 vertices shown in Fig.\ \ref{fig:coutinho}, originally due to Coutinho \cite{coutinho2016a}.
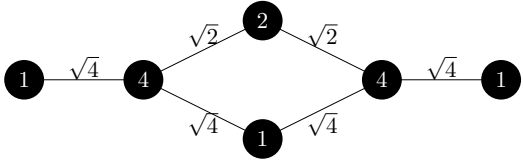
\begin{figure}
\begin{center}
\begin{adjustbox}{width=0.4\textwidth}
\begin{tikzpicture}
\draw (2,0) -- (4,-1) node[circle,style={fill=black,minimum width=1,text=white}]{$1$} -- (6,0);
\draw (0,0) node[circle,style={fill=black,minimum width=1,text=white}]{$1$} -- (2,0) node[circle,style={fill=black,minimum width=1,text=white}]{$4$} -- (4,1) node[circle,style={fill=black,minimum width=1,text=white}]{$2$} -- (6,0) node[circle,style={fill=black,minimum width=1,text=white}]{$4$} -- (8,0) node[circle,style={fill=black,minimum width=1,text=white}]{$1$};
\draw (1,0.2) node{$\sqrt{4}$};
\draw (3,0.75) node{$\sqrt{2}$};
\draw (5,0.75) node{$\sqrt{2}$};
\draw (3,-0.75) node{$\sqrt{4}$};
\draw (5,-0.75) node{$\sqrt{4}$};
\draw (7,0.2) node{$\sqrt{4}$};
\end{tikzpicture}
\end{adjustbox}
\end{center}
\caption{A particularly small perfect transfer graph of distance 4, initially due to Coutinho \cite{coutinho2016a}.}\label{fig:coutinho}
\end{figure}

One way to understand how this example works is that our partitioned graph is entirely agnostic to the specific way that vertices in the two nodes are connected up. This is what \cite{christandl2005} referred to as a ``scrambled hypercube''. We can choose to group these so that there's some extra structure (for the simple case below, it doesn't matter how you scramble it).
\begin{center}
\begin{adjustbox}{width=0.4\textwidth}
\begin{tikzpicture}
\draw (2,0) node[circle,style={fill=black,minimum width=1,text=white}]{$4$} -- (4,0) node[circle,style={fill=black,minimum width=1,text=white}]{$6$} -- (6,0) node[circle,style={fill=black,minimum width=1,text=white}]{$4$};
\draw (7,0) node{$\equiv$};
\fill[black] (8,1.5) circle(0.3);
\fill[black] (8,0.5) circle(0.3);
\fill[black] (8,-0.5) circle(0.3);
\fill[black] (8,-1.5) circle(0.3);

\fill[black] (10,2.5) circle(0.3);
\fill[black] (10,1.5) circle(0.3);
\fill[black] (10,0.5) circle(0.3);
\fill[black] (10,-0.5) circle(0.3);
\fill[black] (10,-1.5) circle(0.3);
\fill[black] (10,-2.5) circle(0.3);

\fill[black] (12,1.5) circle(0.3);
\fill[black] (12,0.5) circle(0.3);
\fill[black] (12,-0.5) circle(0.3);
\fill[black] (12,-1.5) circle(0.3);

\draw [thick] (8,1.5) -- (10,2.5) -- (8,0.5) -- (10,-0.5) -- (8,-0.5) -- (10,1.5) -- (8,1.5) -- (10,0.5) -- (8,-1.5) -- (10,-2.5) -- (8,-0.5); 
\draw [thick] (8,0.5) -- (10,-1.5) -- (8,-1.5);

\draw [thick] (12,1.5) -- (10,2.5) -- (12,0.5) -- (10,-0.5) -- (12,-0.5) -- (10,1.5) -- (12,1.5) -- (10,0.5) -- (12,-1.5) -- (10,-2.5) -- (12,-0.5); 
\draw [thick] (12,0.5) -- (10,-1.5) -- (12,-1.5);
\end{tikzpicture}
\end{adjustbox}
\end{center}
In the central column, we can partition the nodes, merging together the top and bottom vertices in one node (because, between them, they have degree 2 going each way, and the neighbouring node has degree 1), and the other vertices in a single node,
\begin{center}
\begin{adjustbox}{width=0.2\textwidth}
\begin{tikzpicture}
\draw [thick] (2,0) node[circle,style={fill=black,minimum width=1,text=white}]{$4$} -- (4,1) node[circle,style={fill=black,minimum width=1,text=white}]{$2$} -- (6,0) node[circle,style={fill=black,minimum width=1,text=white}]{$4$} -- (4,-1) node[circle,style={fill=black,minimum width=1,text=white}]{$4$} -- (2,0);
\draw (2.3,0.4) node{$1$};
\draw (2.3,-0.4) node{$2$};
\draw (3.6,1) node{$2$};
\draw (3.7,-0.6) node{$2$};
\end{tikzpicture}
\end{adjustbox}
\end{center}
Applying our manipulation rule to the lower node yields:
\begin{center}
\begin{adjustbox}{width=0.2\textwidth}
\begin{tikzpicture}
\draw [thick] (2,0) node[circle,style={fill=black,minimum width=1,text=white}]{$4$} -- (4,1) node[circle,style={fill=black,minimum width=1,text=white}]{$2$} -- (6,0) node[circle,style={fill=black,minimum width=1,text=white}]{$4$} -- (4,-1) node[circle,style={fill=black,minimum width=1,text=white}]{$1$} -- (2,0);
\draw (2.3,0.4) node{$1$};
\draw (2.3,-0.4) node{$1$};
\draw (3.6,1) node{$2$};
\draw (3.7,-0.6) node{$4$};
\end{tikzpicture}
\end{adjustbox}
\end{center}

\begin{lemma}\label{lem:explicit}
For $\Delta=4$, transfer over distances 4 and 5, the example of Coutinho, and its Cartesian product with $P_2$, are the graphs with the minimum number of vertices such that the graph still has a graph quotient coinciding with the standard perfect state transfer chain.
\end{lemma}
\begin{proof}
We start from the engineered chain that we want as the graph quotient, knowing that the first and last vertices must correspond to single qubits. For that to be possible, the second and penultimate qubits must correspond to a number of vertices in the graph equal to the transfer distance. These statements apply to all possible distances. Now we must work out what the possible correspondences are for the remaining column(s). Due to symmetry, both cases effectively only have one column to take care of. We consider the coupling between adjacent columns. Let $\tilde A\in\{0,1\}^{k\times D}$, where $k$ is to be determined. We require that $\tilde A^T\tilde A$ has an all-ones eigenvector with eigenvalue $J^2$, where $J=\sqrt{2(D-1)}$ is the appropriate coupling strength in the quotient graph. It is straightforward to check by brute force that for $D=4$, $k<3$ is not admissible. Similarly for $D=5$, $k<7$ is not admissible. The smallest remaining values of $k$ match the given constructions. 
\end{proof}

This is important because it shows that, for these transfer distances, it is the lower bounds that are lacking, not the quality of the construction. The best known lower bounds (in terms of scaling) on vertex number are polynomial in $D$ \cite{coutinho2018}. In fact, those bounds can be slightly improved in the present context of extremal transfer:
$$
2m=\text{Tr}(A^2)\geq\sum_{\lambda\in\Phi_a}\lambda^2\geq \min_q\sum_{n=1}^{\epsilon_A+1}\left(\frac{2\pi n}{\tau}-q\right)^2
$$
where $m$ is the number of edges in the graph, and $2t_0=\tau\leq 2\pi$ is the perfect revival time. We know that every eigenvalue is separated by an integer multiple of $2\pi/\tau$, and $q$ is just an offset that we can optimise over. Hence,
$$
2m\geq \frac{\epsilon_A(\epsilon_A+1)(\epsilon_A+2)}{12}\geq \frac{D(D+1)(D+2)}{12},
$$
and the number of vertices $N$, is related by $\binom{N}{2}\geq m$ if we want to be completely general. For a fixed maximum degree, we have $Nd\geq 2m$, and we generally expect $d\sim D$. Note that this bound is so weak that for $D=4$, all it tells us is that $N\geq 3$, but clearly $N\geq 5$ to achieve distance 4!.


\section{The \texorpdfstring{$P_3$}{P3} Hypercube Partition}\label{sec:P3}

\subsection{A First Example} \label{sec:first}

Our goal is to reduce the vertex count by considering other equitable distance partitions. In particular, we will take inspiration from the structure of the $P_3$ hypercubes (dimension $D/2$ for even $D$). We will start by explicitly considering the simplest such case, where $D=4$.
\begin{center}
\begin{adjustbox}{width=0.15\textwidth}
\begin{tikzpicture}
\foreach \x in {1,...,3} {
  \draw[black,thick] (1.5*\x,1.5) -- (1.5*\x,4.5);
  \draw[black,thick] (1.5,1.5*\x) -- (4.5,1.5*\x);
  \foreach \y in {1,...,3}
    {
    	\fill[black] (1.5*\x cm,1.5*\y) circle(0.5);
    }
}
\end{tikzpicture}
\end{adjustbox}
\end{center}
Let us take the bottom left-hand corner as the input vertex, and the top right-hand corner as the output vertex. The vertices at distance 1 from the input are all equivalent, as are those at distance 3. At distance 2, the central vertex is distinct from the two corner vertices. Thus, we have a partitioning of the form
\begin{center}
\begin{adjustbox}{width=0.15\textwidth}
\begin{tikzpicture}
\foreach \x in {1,...,3} {
  \draw[black,thick] (1.5*\x-3,1.5) -- (1.5*\x-3,1.5*\x-3);
  \draw[black,thick] (-1.5,1.5*\x-3) -- (1.5*\x-3,1.5*\x-3);
  \foreach \y in {1,...,\x}
    {
    	\fill[black] (3cm-1.5*\x cm,3-1.5*\y) circle(0.5);
    }
}
\matrix[matrix of nodes,nodes={inner sep=0pt,text width=1cm,align=center,minimum height=1cm,white},row sep=0.5cm, column sep=0.5cm]{
 1 & 2 & 1 \\
 2 & 2 \\
 1 \\};
\end{tikzpicture}
\end{adjustbox}
\end{center}
There are no advantages to be gained by applying manipulation rules. However, the $\Delta$-doubling lift returns Fig.\ \ref{fig:coutinho}.

\subsection{General Case}

We shall now describe the general version of this construction. For a $D/2$-dimensional hypercube of $P_3$, each vertex can be labelled by a value $x\in\{0,1,2\}^{D/2}$, where $(0,0,\ldots 0)$ is the input vertex and $(2,2,\ldots,2)$ is the output vertex. The distance from the input vertex is given by $w_x=\sum_{i=1}^{D/2}x_i$. We proceed by observing that all vertices with the same number of 0s in their label, $n_0$, and 2s, $n_2$, are equivalent and can thus be grouped together in a node of the partition. There are
\begin{equation}
\frac{\frac{D}{2}!}{n_0!n_2!(D/2-n_0-n_2)!}	\label{eqn:occupancy}
\end{equation}
such vertices in each node. These have clearly identifiable connectivities. For example, every vertex in the $(n_0,n_2)$ node connects to $n_0$ vertices in the $(n_0-1,n_2)$ node (as there are $n_0$ choices of which 0 to turn into a 1), and $\frac{D}2-n_0-n_2$ vertices in the $(n_0,n_2+1)$ node (changing a 1 into a 2). We thus make the observation that the nodes form a graph that corresponds to half a square lattice, and along each row and column, the corresponding graph quotient has coupling strengths equal to those of the perfect transfer chain of the appropriate length!
\begin{figure*}
\begin{adjustbox}{width=0.8\textwidth}
\begin{tikzpicture}
\foreach \x in {1,...,17} {
  \draw[black,thick] (-12,13.5-1.5*\x) -- (13.5-1.5*\x,13.5-1.5*\x);
  \draw[black,thick] (13.5-1.5*\x,12) -- (13.5-1.5*\x,13.5-1.5*\x);
  \foreach \y in {1,...,\x}
    {
    	\fill[black] (13.5cm-1.5*\x cm,13.5-1.5*\y) circle(0.5);
    }
}
\matrix[matrix of nodes,nodes={inner sep=0pt,text width=1cm,align=center,minimum height=1cm,white},row sep=0.5cm, column sep=0.5cm]{
 4 & 4 & 30 & 35 & 455 & 273 & 2002 & 715 & 1430 & 715 & 2002 & 273 & 455 & 140 & 120 & 16 & 1 \\
 4 & 60 & 105 & 455 & 1365 & 3003 & 5005 & 715 & 715 & 5005 & 3003 & 1365 & 1820 & 420 & 240 & 16  \\
 30 & 105 & 2730 & 2730 & 30030 & 15015 & 10010 & 715 & 10010 & 15015 & 30030 & 2730 & 2730 & 420 & 120  \\
 35 & 455 & 2730 & 10010 & 1001 & 5005 & 15015 & 15015 & 5005 & 1001 & 10010 & 2730 & 1820 & 140  \\
 455 & 1365 & 30030 & 1001 & 1001 & 10010 & 2145 & 10010 & 1001 & 1001 & 30030 & 1365 & 455  \\
 273 & 3003 & 15015 & 5005 & 10010 & 286 & 286 & 10010 & 5005 & 15015 & 3003 & 273  \\
 2002 & 5005 & 10010 & 15015 & 2145 & 286 & 2145 & 15015 & 10010 & 5005 & 2002  \\
 715 & 715 & 715 & 15015 & 10010 & 10010 & 15015 & 715 & 715 & 715  \\
 1430 & 715 & 10010 & 5005 & 1001 & 5005 & 10010 & 715 & 1430  \\
 715 & 5005 & 15015 & 1001 & 1001 & 15015 & 5005 & 715  \\
 2002 & 3003 & 30030 & 10010 & 30030 & 3003 & 2002  \\
 273 & 1365 & 2730 & 2730 & 1365 & 273  \\
 455 & 1820 & 2730 & 1820 & 455  \\
 140 & 420 & 420 & 140  \\
 120 & 240 & 120 \\
 16 & 16  \\
 1 \\};
\end{tikzpicture}
\end{adjustbox}
\caption{The $\Delta=2$, distance 32 transfer graph. These numbers are the number of vertices. Connections are with nearest neighbours of the grid, with degrees determined by the effective coupling strength of the edge: $J=\sqrt{n(k-n)}$ between vertices $n,n+1$ in a row/column of $k$ nodes: the product of the in and out degrees along an edge is equal to $J^2$, and, for a given edge, the total number of incoming connections is equal to the total number of outgoing ones.}\label{fig:crowningglory}
\end{figure*}
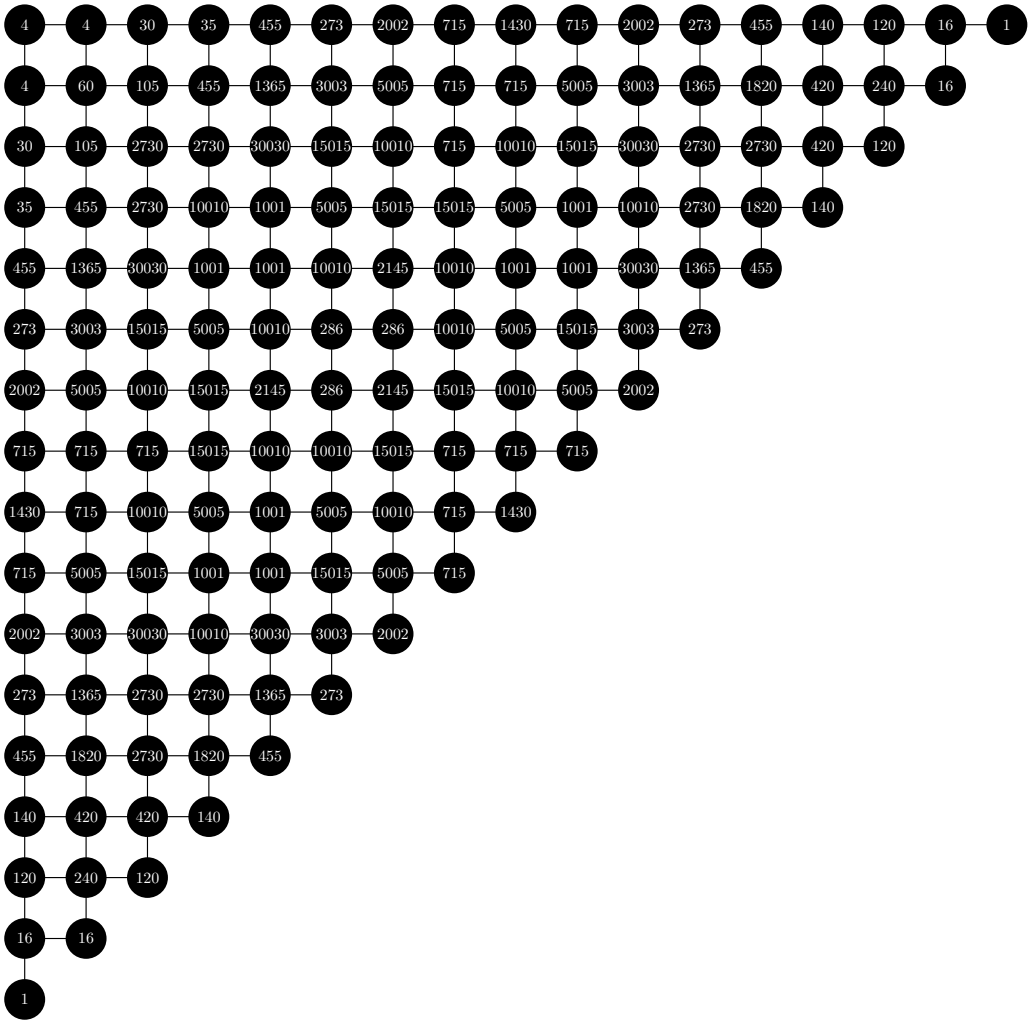
To demonstrate this on a modestly sized graph, consider $D=10$. The node structure is
\begin{center}
\begin{adjustbox}{width=0.48\textwidth}
\begin{tikzpicture}
\foreach \x in {1,...,6} {
  \draw[black,thick] (-3.75,5.25-1.5*\x) -- (5.25-1.5*\x,5.25-1.5*\x);
  \draw[black,thick] (5.25-1.5*\x,3.75) -- (5.25-1.5*\x,5.25-1.5*\x);
  \foreach \y in {1,...,\x}
    {
    	\fill[black] (5.25cm-1.5*\x cm,5.25-1.5*\y) circle(0.5);
    }
}
\matrix[matrix of nodes,nodes={inner sep=0pt,text width=1cm,align=center,minimum height=1cm,white},row sep=0.5cm, column sep=0.5cm]{
 1 & 5 & 10 & 10 & 5 & 1 \\
 5 & 20 & 30 & 20 & 5 \\
 10 & 30 & 30 & 10 \\
 10 & 20 & 10 \\
 5 & 5 \\
 1 \\};

 \foreach \x in {1,...,6} {
  \draw[black,thick] (5,5.25-1.5*\x) -- (14-1.5*\x,5.25-1.5*\x);
  \draw[black,thick] (14-1.5*\x,3.75) -- (14-1.5*\x,5.25-1.5*\x);
  \foreach \y in {1,...,\x}
    {
    	\fill[black] (14cm-1.5*\x cm,5.25-1.5*\y) circle(0.2);
    }
}
\matrix[matrix of nodes,nodes={inner sep=0pt,text width=1cm,align=center,minimum height=1cm,black},row sep=0.5cm, column sep=0.5cm] at (8.75,1) {
 $\sqrt{5\cdot 1}$ & $\sqrt{4\cdot 2}$ & $\sqrt{3\cdot 3}$ & $\sqrt{2\cdot 4}$ & $\sqrt{1\cdot 5}$ \\
 $\sqrt{4\cdot 1}$ & $\sqrt{3\cdot 2}$ & $\sqrt{2\cdot 3}$ & $\sqrt{1\cdot 4}$ \\
 $\sqrt{3\cdot 1}$ & $\sqrt{2\cdot 2}$ & $\sqrt{1\cdot 3}$ \\
 $\sqrt{2\cdot 1}$ & $\sqrt{1\cdot 2}$ \\
 1 \\};
\end{tikzpicture}
\end{adjustbox}
\end{center}
The vertices with occupancy 20 can be reduced to 5 by the manipulation rule. To the side, the graph quotient is depicted. The effective coupling strengths are shown for the horizontal couplings. Nodes are arranged such that $n_0$ decreases vertically, and $n_2$ increases horizontally, meaning that state transfer is from bottom left $(n_0,n_2)=(5,0)$ to top right, $(0,5)$. The $\Delta=2$ graph contains 198 vertices, while the $\Delta=4$ version requires 199 vertices.

To have the best chance of the manipulation rule being applicable, we want multiple factors of 2 to appear in the node occupancies given by Eq.\ (\ref{eqn:occupancy}). If $D$ is a power of 2, there are only 6 nodes whose occupancies are not divisible by 4: $(n_0,n_2)=(D/2,0),(0,D/2),(0,0),(D/4,0),(0,D/4),(D/4,D/4)$. Modulo a few edge effects due to the need for even degrees, we can therefore apply the grouped-node version of the manipulation rule on almost the entire graph! Moreover, if either $n_0$ or $n_2$ is odd, the final occupancy is still even, which is perfect for the $\Delta$-doubling lift! With some additional manual manipulation, we have achieved particularly drastic reductions for the $D=16,32$ cases. The distance 32, $\Delta=2$ transfer graph has 680913 vertices, an efficiency of $\eta=0.606$, and is depicted in Fig.\ \ref{fig:crowningglory}. The $\Delta$-doubling lift yields a $\Delta=4$ version with 830895 vertices, an efficiency of $\eta=0.615$.

Applying the node-splitting technique to the $D=32$ perfect transfer graph of Fig.\ \ref{fig:crowningglory} yields no improvement. For the $\Delta=4$ case, there is some minor improvement as the nodes at $(n_0,n_2)=(0,8)$, $(8,0)$ and $(8,8)$, each with node occupancy 1430, can be replaced by two nodes with occupancies 286 and 130. The total vertex count is therefore reduced to 827853, efficiency $\eta=0.614$.

\subsection{Sequential Improvements}

The Cartesian product $G^{\square k}$ of $G$ gives transfer over distances $kD$ with identical efficiency. However, can we do better? As we increase the transfer distance, can we reduce the vertex number compared to $|G|^k$? We will now give one mechanism that provides this, and estimate the revised efficiency that it gives.

For a set of nodes labelled by $i$ that make up graph $G$, $G^{\square 2}$ has a corresponding set of nodes labelled by $(i,j)$. Provided $i\neq j$, the nodes $(i,j)$ and $(j,i)$ are equivalent and can be replaced by a single node with double the occupancy. Nodes $(i,i)$ retain the same occupancy but their degrees double. If we repeat this to get $G^{\square 4}$, but with the new node structure, then the vast majority of nodes now contain a factor of 4 in their occupancies. We apply the manipulation rule on most of these vertices, reducing the vertex number. If we assume that the efficiency of $G$, with transfer distance $D$ was $\eta_D$, then the new efficiency is approximately
$$
\eta_{4D}\gtrsim\frac{1}{4D}\log_2\frac{2^{4D\eta_D}}{4}=\eta_D-\frac{1}{2D}.
$$
Then we repeat this to give
$$
\eta_{4^qD}\gtrsim\eta_D-\frac{2}{D}\sum_{k=1}^q\frac{1}{4^k}.
$$
In the large $q$ limit, we have
$$
\eta\rightarrow\eta_D-\frac{2}{3D}.
$$
This lets us estimate achievable efficiencies of about $\eta=0.584$ for even $D$ and $\eta=0.594$ for odd $D$, starting from the $D=32$ case.

\section{Conclusions \& Future Work}

We can achieve perfect state transfer over distance $D$ with a vertex count that scales as $2^{\eta D}$ where $\eta=0.606$ if $D$ is even, and $\eta=0.614$ if $D$ is odd. We have estimated that better efficiencies are achievable in a relatively straightforward manner, but for much larger graphs. We also know that if we wish to transfer over an odd distance, $D\leq d$, which is severely limiting for practical applications.


The big question that remains open is whether or not an exponential scaling in the number of vertices is necessary and, if so, can one bound from below the exponent? 
 Clearly the existing lower bounds are too weak because there are already large gaps between those bounds and the transfer graphs of distance 1 to 5 that we have shown to be optimal. One clear challenge is that existing lower bound techniques only take into account the space $\Phi_a$. Since we achieve effective coupling strengths other than unity in the quotient graph by having node occupancies larger than 1, there is generally a large space outside $\Phi_a$ whose size we currently have no way of estimating.

It helps to illustrate the difficulty in resolving the polynomial/exponential divide in vertex number by considering a family of graphs due to Stevanovi\'c \cite{Note1}. In our graph partition terminology, an example (from which the family may readily be extrapolated by looking at the pattern of alternate nodes) is

\begin{center}
\begin{adjustbox}{width=0.5\textwidth}
\begin{tikzpicture}
\draw (0,0) node[circle,style={fill=black,minimum width=1,text=white}]{$5$} -- (1.5,0) node[circle,style={fill=black,minimum width=1,text=white}]{$1$} -- (3,0) node[circle,style={fill=black,minimum width=1,text=white}]{$4$} -- (4.5,0) node[circle,style={fill=black,minimum width=1,text=white}]{$2$} -- (6,0) node[circle,style={fill=black,minimum width=1,text=white}]{$3$} -- (7.5,0) node[circle,style={fill=black,minimum width=1,text=white}]{$3$} -- (9,0) node[circle,style={fill=black,minimum width=1,text=white}]{$2$} -- (10.5,0) node[circle,style={fill=black,minimum width=1,text=white}]{$4$} -- (12,0) node[circle,style={fill=black,minimum width=1,text=white}]{$1$} -- (13.5,0) node[circle,style={fill=black,minimum width=1,text=white}]{$5$};
\draw (0.4,0.2) node{$1$};
\draw (1.1,0.2) node{$5$};
\draw (1.9,0.2) node{$4$};
\draw (2.6,0.2) node{$1$};
\draw (3.4,0.2) node{$2$};
\draw (4.1,0.2) node{$4$};
\draw (4.9,0.2) node{$3$};
\draw (5.6,0.2) node{$2$};
\draw (6.4,0.2) node{$3$};
\draw (7.1,0.2) node{$3$};
\draw (7.9,0.2) node{$2$};
\draw (8.6,0.2) node{$3$};
\draw (9.4,0.2) node{$4$};
\draw (10.1,0.2) node{$2$};
\draw (10.9,0.2) node{$1$};
\draw (11.6,0.2) node{$4$};
\draw (12.4,0.2) node{$5$};
\draw (13.1,0.2) node{$1$};
\end{tikzpicture}
\end{adjustbox}
\end{center}
These graphs exhibit perfect revivals, but do not exhibit perfect transfer, only because instead of every neighbouring gap in eigenvalues being an odd integer, there is a single gap that is an even integer -- our example above has eigenvalues $(\pm1,\pm2,\pm3,\pm4,\pm5)$, while a chain with spectrum $(0,\pm1,\pm2,\pm3,\pm4,\pm5)$ has perfect transfer. Otherwise, perfect transfer would be possible, and the construction only requires a number of vertices $N=\frac14(D+1)(D+3)$ that, crucially, is only quadratic in $D$.

Furthermore, we have concentrated exclusively on graphs whose ultimate quotient is just the perfect state transfer chain with engineered couplings $\sqrt{n(N-n)}$. There are myriad other analytic options \cite{albanese2004}. The basic constructions that yield these will typically involve more vertices, but it may be the case that there are multiple common factors, leaving them more amenable to the reduction techniques investigated here. The Stevanovi\'c family of graphs may provide insight here -- they are equivalent to chains with coupling strengths ($J_{2n}=\sqrt{n(N/2-n)}$ and $J_{2n+1}=\sqrt{(n+1)(N/2-n)})$ that have significant factors in common between consecutive edges, and we believe that this is responsible for facilitating such a massive reduction in vertex number. However, our experiments to try and replicate aspects of this structure have all resulted in graphs that do not possess perfect state transfer.

{\em Acknowledgements}: We would like to thank G.\ Coutinho for useful conversations. This work was supported by EPSRC grant EP/N035097/1, and partially conducted during the Algebraic Graph Theory \& Quantum Walks at the University of Waterloo, April 2018.

%

\end{document}